\documentclass[11pt]{article}
\usepackage[letterpaper, margin=1in]{geometry}
\usepackage{pdfpages}

\usepackage[T1]{fontenc} 
\usepackage[english]{babel}

\usepackage[numbers]{natbib}
\usepackage{xspace}
\usepackage{algorithm}
\usepackage{algpseudocode}
\usepackage{amsmath}
\usepackage{color}
\usepackage{graphicx}
\usepackage{tikz}
\usepackage[colorlinks=true, allcolors=blue]{hyperref}

\usepackage[utf8]{inputenc} 
\usepackage{url}            
\usepackage{booktabs}       
\usepackage{amsfonts}       
\usepackage{nicefrac}       
\usepackage{microtype}      

\usepackage{amsthm}
\usepackage{amsmath,amssymb}
\usepackage{comment}        
\usepackage{enumitem}  
\usepackage{parskip}
\usepackage{cleveref} 

\newtheorem{theorem}{Theorem}[section]
\newtheorem*{theorem*}{Theorem}
\newtheorem{proposition}[theorem]{Proposition}
\newtheorem{remark}[theorem]{Remark}
\newtheorem*{remark*}{Remark}
\newtheorem{lemma}[theorem]{Lemma}

\newcommand{\bz}{{\mathbf z}}

\newcommand{\be}{\boldsymbol e}

\newcommand{\Acal}{\mathcal{A}}
\newcommand{\Bcal}{\mathcal{B}}

\newcommand{\Lcal}{\mathcal{L}}

\newcommand{\Ocal}{\mathcal{O}}

\newcommand{\Tcal}{\mathcal{T}}

\newcommand{\Xcal}{\mathcal{X}}
\newcommand{\Ycal}{\mathcal{Y}}

\newcommand{\Ebb}{\mathbb{E}}
\newcommand{\Nbb}{\mathbb{N}}

\newcommand{\Rbb}{\mathbb{R}}

\newcommand{\bx}{{\boldsymbol x}}
\newcommand{\bdelta}{{\boldsymbol \delta}}
\newcommand{\by}{{\boldsymbol y}}

\usepackage{bbm}
\newcommand{\one}{\mathbbm{1}}
\newcommand{\oneb}{\boldsymbol{\mathbbm{1}}}

\newcommand{\OT}{\text{OT}}
\newcommand{\Cost}{{\normalfont \texttt{Cost}}}
\newcommand{\Moves}{{\normalfont \texttt{Moves}}}
\newcommand{\Runtime}{{\normalfont \texttt{Runtime}}}
\newcommand{\LCA}{\text{A}}

\newcommand{\ceil}[1]{{\lceil #1 \rceil}}
\newcommand{\floor}[1]{{\lfloor #1 \rfloor}}

\DeclareMathOperator*{\argmin}{arg\,min}

\newcommand{\cL}{\mathcal {L}}

\title{Collective Tree Exploration via Potential Function Method}
\author{
  Romain Cosson\\
  \small{\texttt{romain.cosson@inria.fr}}
  \and
  Laurent Massoulié\\
  \small{\texttt{laurent.massoulie@inria.fr}}
}
\date{}
\begin{document}
\maketitle

\begin{abstract}
We study the problem of collective tree exploration (CTE) where a team of $k$ agents is tasked to traverse all the edges of an unknown tree as fast as possible, assuming complete communication between the agents \cite{fraigniaud2006collective}. 
In this paper, we present an algorithm performing collective tree exploration in only $2n/k+\Ocal(kD)$ rounds, where $n$ is the number of nodes in the tree, and $D$ is the tree depth. 
This leads to a competitive ratio of $\Ocal(\sqrt{k})$ for collective tree exploration, the first polynomial improvement over the $\Ocal(k)$ ratio of depth-first search. 
Our analysis relies on a game with robots at the leaves of a continuously growing tree, which is presented in a similar manner as the `evolving tree game' of \cite{bubeck2022shortest}, though its analysis and applications differ significantly. 
This game extends the `tree-mining game' (TM) of \cite{cosson2023breaking} and leads to guarantees for an asynchronous extension of collective tree exploration (ACTE).
Another surprising consequence of our results is the existence of algorithms $\{\Acal_k\}_{k\in \Nbb}$ for layered tree traversal (LTT) with cost at most $2L/k+\Ocal(kD)$, where $L$ is the sum of edge lengths and $D$ is the tree depth. For the case of layered trees of width $w$ and unit edge lengths, our guarantee is thus in $\Ocal(\sqrt{w}D)$. 
\end{abstract}

\section{Introduction}
The present study concerns  collaborative tree exploration (CTE), a problem introduced in the field of distributed computing by \citep{fraigniaud2004collective}. Its goal is for a team of agents or robots, initially located at the root, to explore an unweighted tree, going through all edges as quickly as possible before returning to the root. At all rounds, each robot moves along one edge to reach a neighboring node. When a robot attains the endpoint of an unknown edge, the existence of that edge
is revealed to the team. 
Following the centralized full-communication setting, we assume that robots can communicate and compute at no cost. They thus share at all times a common map of the sub-tree that has already been explored and of the discovered edges that have not yet been traversed.

Another seemingly unrelated problem is layered graph traversal (LGT). It was introduced in the literature on online algorithms by \cite{papadimitriou1991shortest,fiat1998competitive}. We describe the problem for the special case of trees, which was extensively studied \cite{fiat1998competitive, ramesh1993traversing, bubeck2022shortest}. The goal is for a single agent, initially located at the origin of an unknown tree with weighted edge lengths, to reach some other node called the target. The tree is revealed iteratively: at step $i$ the agent must move to a node in the $i$-th layer, i.e. at combinatorial depth $i$, using only the edges revealed in preceding layers. The cost of a randomized algorithm for layered tree traversal is the expected distance travelled by the searcher to attain the target, located in the last layer. 

\paragraph{Main results.} In this paper we present an algorithm performing collective tree exploration (CTE) with $k$ robots in $2n/k+\Ocal(kD)$ synchronous rounds for any tree with $n$ nodes and depth $D$. This algorithm, when used by $\sqrt{k}$ robots, thus achieves a competitive ratio of order $\Ocal(\sqrt{k})$. The algorithm can also be adapted to an asynchronous setting (ACTE) defined in Section \ref{sec:reductions}. 

Our analysis relies on a two-player game, that we call the `continuous tree-mining game' (CTM). In this game, the adversary controls the continuous evolution of a tree while the player controls the position of $k$ `miners' located at its leaves. The game differs from the `evolving tree game' of \cite{bubeck2022shortest} in that the player may block the extension of a leaf of the tree by attributing it a single miner. We show that it is possible for the player of this game to get all miners to reach depth $D$ with a total movement cost of at most $\Ocal(k^2D)$.

Another consequence of our analysis is a sequence of algorithms $\{\Acal_k\}_{k\in \Nbb}$ for layered tree traversal (LTT) which satisfy for any $k\in \Nbb$ that the expected runtime of algorithm $\Acal_k$ on a tree $T$ of depth $D$ and length $L$ is bounded by $2L/k+\Ocal(kD)$. We highlight for the first time a correspondence between collective tree exploration and layered graph traversal, which have until now been introduced and studied by different communities.

\paragraph{Background on Collective Tree Exploration.} The problem of collective tree exploration has a rich history in the field of distributed algorithms and robotics. It was introduced by \cite{fraigniaud2004collective} along with two communication models. A centralized `complete communication' model, in which communications are unrestricted, which we study in this paper ; and a distributed `write-read communication' model in which agents communicate through whiteboards located at all nodes. A collaborative exploration algorithm is said to be order $c(k)-$competitive if its runtime on a tree with $n$ nodes and of depth $D$ is bounded by $\Ocal\left(c(k)\left(\frac{n}{k}+D\right)\right)$ (see \cite{fraigniaud2006collective}). A competitive ratio of $\Ocal(k)$ is thus trivially achieved by a single depth-first search. \cite{fraigniaud2006collective} proposed a $\Ocal(k/\log(k))$-competitive algorithm which can be implemented in the distributed communication model, and thus also in the complete communication model. They suggested that a constant ratio could be achieved in the complete communication model (see initial version \cite{fraigniaud2004collective}). This conjecture was disproved by \cite{dynia2007robots} who showed that the competitive ratio of any deterministic algorithm is at least in $\Omega(\log(k)/\log\log(k))$. 
Many works followed, tackling diverse questions such as: quasi-linear algorithms \cite{ortolf2014recursive}, power constraints \cite{dynia2006power}, and the case of many explorers $k\gg n$ \cite{dereniowski2015fast}. 
A new type of competitive analysis was proposed by \cite{brass2011multirobot}, with a guarantee of the form $2n/k+\Ocal((D+k)^k)$, where $2n/k$ is a lower-bound on the time required by the robots to traverse all edges and return to the origin and where we can thus call the quantity $\Ocal((D+k)^k)$ the `competitive overhead' or `penalty'. 
This quantity was improved to $\Ocal(\log(k)D^2)$ by a simple algorithm combining breadth-first search and depth-first search \cite{cosson2023breadth}. The approach was then recently used by  \cite{cosson2023breaking} to slightly improve the competitive ratio of collective tree exploration in the complete communication model to order $\Ocal(k/\exp(\sqrt{\ln 2\ln k}))$, using a linear in-$D$ competitive overhead of $\Ocal(k^{\log_2(k)-1}D)$. These two recent analyses were performed in the complete communication model for an asynchronous extension of collective tree exploration (ACTE). The competitive ratio presented in this paper, in $\Ocal(\sqrt{k})$, is the first to display a polynomial improvement over the aforementioned $\Ocal(k)$ ratio of depth-first search. The present discussion is summarized in Table \ref{table:1}.

\renewcommand{\arraystretch}{1.3}
\setlength{\tabcolsep}{20pt}
\begin{table}[h!]
\centering
\begin{tabular}{||c | c c||} 
 \hline
  $\Ocal(\cdot)$& Competitive Ratio $c(\cdot)$ & Competitive Overhead $f(\cdot,\cdot)$\\
  Runtime  & $c(k)(\frac{n}{k}+D)$ & $\frac{2n}{k}+f(k,D)$ \\[0.5ex] 
 \hline\hline
 \cite{fraigniaud2004collective} & $k/\ln k$ & - \\
 \cite{brass2011multirobot} & - & $(D+k)^k$\\
 \cite{cosson2023breadth} & - & $\ln kD^2$\\
 \cite{cosson2023breaking} & $k/\exp(\sqrt{\ln 2 \ln k})$ & $k^{\log_2 k-1}D$\\ 
 This work & $\sqrt{k}$ & $kD$\\[1ex] 
 \hline
\end{tabular}
\caption{Previous work on competitive analysis of collaborative tree exploration (CTE), in the complete communication model. All results hold up to a multiplicative constant.}
\label{table:1}
\end{table}

\paragraph{Background on Layered Graph Traversal.} The problem has a rich history in the field of online algorithms. It was first described in a paper by Papadimitriou and Yannanakis \cite{papadimitriou1991shortest} titled `Shortest path without a map' analyzing the case of layered graph of width $w=2$, where the width is defined as the maximum cardinality of a layer. Independently the problem was introduced by \cite{chrobak1991server} under the denomination `Metrical Service System', highlighting its connection to so-called `Metrical Task Systems' \cite{borodin1992optimal}.
The equivalence, between both settings, was noticed by \cite{fiat1998competitive}. They also introduced the denomination Layered Graph Traversal (LGT) and observed that its competitive analysis can be reduced to the special case of Layered Tree Traversal (LTT). For arbitrary width $w$ and depth $D$, they proposed a deterministic algorithm with cost bounded by $\Ocal(9^wD)$, i.e. of competitive ratio $\Ocal(9^w)$. 
The quantity was later improved to $\Ocal(w2^w)$ by \cite{burley1996traversing}, nearly matching the lower-bound in $\Omega(2^w)$ \cite{fiat1998competitive}. 
Using a randomized algorithm, \cite{bubeck2022shortest} obtained a $\Ocal(w^2)$ competitive ratio, nearly matching the randomized lower bound in $\Omega(w^2/\log(w))$ \cite{ramesh1993traversing}. The approach of \cite{bubeck2022shortest} is to use a two-player game, that they call the `evolving tree game', which shares some similarities with the `continuous tree-mining game' (CTM) presented in this paper, although applications are totally different. In particular, contrarily to the aforementioned guarantees which are all of the form $\Ocal(c(w)D)$, our bound in $\Ocal\left(2L/k+kD\right)$ depends on the sum of all edge lengths $L$ and on the depth $D$ of the tree but not on the width $w$.

\paragraph{Notations and definitions.} The following definitions are used throughout the paper. A \textit{tree} $T=(V,E)$ is a connected acyclic graph. One specific node, called the root, is denoted $r\in V$. Every other node $u \in V\setminus \{r\}$ has a unique parent denoted by $p(u)$. For two nodes $u,v \in V$ we say that $u$ is a descendant of $v$ or equivalently that $v$ is an ancestor of $u$ and we denote by $u\preceq v$ if $v$ can be obtained from $u$ by iterating the parent function $p(\cdot)$. For any two nodes $u,v\in V$ we denote by $\LCA(u,v)$ the lowest common ancestor of $u$ and $v$. We also denote by $u\rightarrow v$ (resp. $u\leftrightarrow v$) the sequence of nodes in the shortest path from $u$ to $v$, excluding $v$ (resp. including $v$). A leaf of $T$ is a node $\ell \in V$ which has no descendant. The set of all leaves of $T$ is denoted by $\Lcal(T)$. We say that a tree is \textit{simple} if no node has degree $2$, except possibly for the root. 

A \textit{weighted} tree is a tree in which edges have a positive length, i.e. $E \subset V\times V \times \Rbb^+ $. An unweighted tree can be seen as a weighted tree where all edge lengths are equal to $1$. For $u\in V \setminus\{r\}$, the length of edge $(u,p(u))$ is denoted $d_u \in \Rbb^+$. For any two nodes $u,v \in V$ we denote by $d(u,v)$ the distance from $u$ to $v$, which can be defined by $d(u,v) = \sum_{w\in u\rightarrow \LCA(u,v)}d_w + \sum_{w \in v \rightarrow \LCA(u,v)}d_w$.

For some integer $k\geq 2$, a \textit{discrete configuration} on a tree $T$ is a collection $\bx\in \Nbb^{\Lcal(T)}$ satisfying $\sum_{\ell\in \Lcal(T)}x_\ell = k$. It can be extended to a collection $\bx \in \Nbb^V$ by setting $\forall u\in V: x_u = \sum_{\ell \preceq u}x_\ell$. The set of all discrete configurations is denoted $\Xcal(T)$. A \textit{fractional configuration} on $T$ is a collection $\by\in \Rbb_+^{\Lcal(T)}$ satisfying $\sum_{\ell\in \Lcal(T)}y_\ell = k$. It can  be extended to a collection $\by\in \Rbb_+^{V}$ by setting $\forall u\in V: y_u = \sum_{\ell\preceq u}y_\ell$. The set of all fractional configurations is denoted $\Ycal(T)$. For any two configurations (discrete or fractional) $\bx$ and $\bx'$, we define the optimal transport cost 
$\OT_T(\bx,\bx') = \sum_{u\in V}d_u|x_u-x'_u|.$

\paragraph{Structure of the paper.} The paper is organized as follows. In Section \ref{sec:analysis} we define and analyze the continuous tree-mining game (CTM). In Section \ref{sec:reductions} we present the reductions that allow to relate it to collective tree exploration (CTE) and layered tree traversal (LTT). In Section \ref{sec:applications} we then apply the results of 
Section \ref{sec:analysis} to obtain new guarantees for both problems. 

\section{Analysis of the continuous tree-mining game}\label{sec:analysis}
In this section, we introduce and analyze a two-player game that we call the continuous tree-mining game (CTM). We first define the game in Section~\ref{sec:game-def}, then we present an algorithm for the player of the game in Section \ref{sec:algorithm}, finally we provide an analysis of the player's algorithm in Section \ref{sec:compet-analysis}. The continuous tree-mining game is tightly connected to the problems of collective tree exploration (CTE) and of layered graph traversal (LGT), as we shall see in Section \ref{sec:reductions}.
\subsection{The continuous tree-mining game}\label{sec:game-def}
The state of the game is defined at any continuous time $t$. It consists of a \textit{simple weighted tree} $T(t)$, the evolution of which is controlled by the adversary, and of a \textit{discrete configuration} $\bx(t)$ over $T(t)$, the evolution of which is controlled by the player. We now precise the actions available to the player and the adversary at each instant.

\paragraph{Adversary.} The adversary can do three things: kill a leaf, give some children to a leaf, or elongate the edge leading to a leaf. The first two operations occur instantaneously, while the other is performed continuously over time. We now detail all three operations.

\textit{Leaf edge elongation:} Between discrete changes to the tree by the adversary, or discrete moves by the player, at any given continuous time $t$ the adversary distinguishes one leaf $\ell(t)$ and lets the length $d_{\ell(t)}$ of the corresponding edge increase at unit rate. The adversary is only allowed to choose for $\ell(t)$ a leaf with more than one robot, i.e. such that $x_{\ell(t)}(t)\geq 2$.

\textit{Forking at a leaf:} At discrete time points the adversary can choose  some leaf $\ell$ hosting a number $x_{\ell}\geq 3$ of robots, and endow this leaf with some number $m\in\{2,\ldots,x_\ell-1\}$ of children. Denoting by $\ell_1,\ldots,\ell_m$ these children, the newly created edges of the form $(\ell,\ell(i))$, $i\in[m]$, are initialised with some length $\delta$, where $0<\delta\leq 1$ is a quantity chosen by the player at the time of the fork.

\textit{Killing a leaf:} At discrete time points the adversary can choose to kill some leaf $\ell$. The corresponding edge $(\ell,p(\ell))$ is then also suppressed from the tree. In case node $p(\ell)$ is distinct from the root $r$, and had only one child $u$ besides $\ell$, we then merge the two edges $(u,p(\ell))$ and $(p(\ell),p(p(\ell)))$ into a single edge $(u,p(p(\ell)))$ (suppressing node $p(\ell)$) and endow this new edge with length $d_{p(\ell)}+d_u$, preserving the distance between $p(p(\ell))$ and $u$ in the new tree. At all times, the tree thus remains \textit{simple}.

\paragraph{Player.} At any instant the player can move robots from one leaf to another one, by changing the configuration $\bx$. The cost of going from configuration $\bx$ to $\bx'$ is equal to $\OT_T(\bx,\bx')$. When the adversary kills a leaf $\ell$, the player is forced to move the corresponding $x_\ell$ robots to other leaves that are still alive and to pay the associated cost. When the adversary forks at some leaf $\ell$, endowing it with $m$ children, the player chooses the length $\delta<1$ of the fork and must assign its $x_\ell$ robots to these newly created $m$ leaves, paying the associated cost. Between two discrete re-allocations of robots, when the adversary elongates the leaf $\ell(t)$, the player continuously incurs cost at rate $x_{\ell(t)}$.

\paragraph{Goal of the player.} A strategy for the player is a continuous-time algorithm which determines the response of the player to any modification of the tree by the adversary. It is formally defined as a function $\bx' = \Acal(T,\bx,T')$ which returns a new configuration $\bx'\in \Xcal(T')$ given the previous state of the game $\bx\in \Xcal(T)$ and a modified tree $T'$. Note that taking $T'=T$ allows to account for continuous elongation. We say that a strategy is $f(k,D)$-bounded, for some real-valued function $f(\cdot,\cdot)$, if it is such that no matter how the adversary plays, the cost incurred by the player is always less than $f(k,D)$, where $D$ denotes the depth of the highest leaf. We will be particularly interested in the case where $f$ is linear in $D$.

The interest of the continuous tree-mining game lies in the following reduction.
\begin{theorem}[\cref{th:rcte,th:reduction-lgt,th:tm-ctm,prop:fractional,th:tm-acte} in Section \ref{sec:reductions}]\label{th:recall}
For any $f(k,D)$-bounded strategy for the player, there is a collective tree exploration (CTE) algorithm $\Acal$ such that for any unweighted tree $T$ with $n$ nodes and depth $D$, 
$$\Runtime(\Acal,T) \leq \frac{2n+f(k,D)}{k}+D+1.$$
Also, there exists a collection of layered tree traversal (LTT) randomized algorithms $\{\Acal_k\}_{k\in \Nbb}$ satisfying for any layered tree $T$ of length $L$ and depth $D$,
$$\Ebb(\Cost(\Acal_k,T))\leq \frac{2L+f(k,D)}{k}+1.$$
\end{theorem}
The goal of the rest of this section is thus to prove the following theorem.
\begin{theorem}[Proposition \ref{prop:potential-arguement} in Section \ref{sec:analysis}]\label{th:main}
    There exists a $f(k,D)$-bounded strategy for the continuous tree-mining game (CTM), with $f(k,D) = \Ocal(k^2D)$.
\end{theorem}

\subsection{A potential-based algorithm}\label{sec:algorithm}
We assign to any configuration $\bx \in \Xcal(T)$ a potential $\Psi(T,\bx)$ defined by
\begin{equation}
\Psi(T,\bx):=\sum_{u\in V\setminus \{r\}}d_u \phi(x_u),
\end{equation}
where $\phi$ is a strongly convex function to be determined later. We then define the strategy of the player by the following equation, 
\begin{equation}\label{eq:dynamics}
    \Acal(T,\bx,T') = \argmin_{\bx'\in \Xcal(T')} \Psi(T',\bx') + \OT_{T'}(\bx,\bx'),
\end{equation}
in which ties are always broken in favor of any $\bx'\neq \bx$. 
Note that Algorithm \eqref{eq:dynamics} thus enforces the constraint that at all times $t$ and for any configuration $\bx'\neq \bx(t)$,
$\Psi(\bx(t))-\Psi(\bx') < \OT_{T(t)}(\bx(t),\bx').$
We start by establishing some desirable properties on the dynamics of this algorithm.

\begin{proposition}[Dynamics of $\bx(t)$ following \eqref{eq:dynamics}]\label{prop_dynamics_x} While the adversary elongates a leaf, the moves of the player are all from the elongated leaf to other leaves of the tree, and no two robots are moved simultaneously. When the adversary deletes a leaf, all moves of the player are from the deleted leaf to other leaves. When the adversary forks a leaf $l$, with $m\leq x_l-1$ children denoted $\{\ell_1,\dots,\ell_m\}$, there is a choice of a small real $\delta>0$ such that the new configuration $\bx'$ satisfies for all previously existing node $u\not\in\{\ell_1,\dots,\ell_m\}:x_u = x'_u$ and for all newly created leaves $\ell\in\{\ell_1,\dots,\ell_m\}:x_\ell \in \{\floor{x_l/m},\ceil{x_l/m}\}$. At all times, the configuration of the game $\bx$ satisfies $\forall \ell\in \Lcal(T): x_\ell\geq 1$.
\end{proposition}

The proof of this theorem relies on the notion of `tension' between configurations. For any two configurations $\bx$ and $\bx'$ we call the tension from $\bx$ to $\bx'$ and we denote by $\tau(\bx\rightarrow \bx')$ the decrease in potential obtained when going from configuration $\bx$ to configuration $\bx'$, i.e.  $\tau(\bx\rightarrow \bx') =\Psi(\bx)- \Psi(\bx')$. For a configuration $\bx$ and two leaves $\ell,\ell'$ we call the tension from $\ell$ to $\ell'$ in $\bx$ and denote by $\tau_\bx(\ell\rightarrow \ell')$ the decrease in potential obtained by displacing a robot from $\ell$ to $\ell'$ in configuration $\bx$, i.e. $\tau_\bx(\ell\rightarrow\ell') = \Psi(\bx) - \Psi(\bx+\be_{\ell'}-\be_{\ell})$. Note that this quantity is only defined if $x_\ell\geq 1$, which will always be the case as stated in Proposition \ref{prop_dynamics_x}. The structure imposed on the potential $\Psi$ then leads to the following lemma, which essentially says that atomic moves (where only one robot moves at a time) are favored over simultaneous moves by Algorithm \eqref{eq:dynamics}. 

\begin{lemma}\label{lemma:tension-decomposition}
    Consider configurations $\bx$ and $\bx'$ such that $||\bx-\bx'||_1 = 2h$ (where $||\bx-\bx'||_1 = \sum_{\ell\in \Lcal(T)}|x_\ell-y_\ell|$ is always even) and consider $\ell_1\rightarrow \ell_1',  \dots,  \ell_h\rightarrow \ell_h'$ an optimal transport plan going from $\bx$ to $\bx'$. The following inequality is always satisfied, 
    $$\tau(\bx\rightarrow \bx') \leq \tau_\bx(\ell_1 \rightarrow \ell_1')+\dots+\tau_\bx(\ell_h\rightarrow \ell_h').$$
    Furthermore, this inequality is strict if there are overlaps in the transport plan, i.e. a pair $i,j$ such that the shortest paths $\ell(i)\rightarrow\ell'(i)$ and $\ell_j \rightarrow\ell_j'$ have an intersection of positive length. 
\end{lemma}
\begin{proof}
    We will show the property by induction on $h$. We observe that the property is true for $h=1$. We now assume that the property is true for some $h\geq 1$ and aim to show it at $h+1$. We consider two leaves $\ell \rightarrow \ell'$ of the transport plan from $\bx$ to $\bx'$ and we define $\bx'' = \bx'-\bz$ with $\bz=\be_{\ell'}-\be_{\ell}$, the configuration where the move $\ell\rightarrow\ell'$ did not take place. Observe that $||\bx-\bx''||_1 = 2h$ will later enable us to apply the induction hypothesis to $\tau(\bx\rightarrow\bx'')$. We decompose as follows, 
    $$\Psi(\bx)-\Psi(\bx') = \sum_{u\not\in \ell\leftrightarrow \ell'}d_u(\phi(x_u)-\phi(x'_u))+ \sum_{u\in \ell\leftrightarrow \ell'}d_u(\phi(x_u)-\phi(x'_u)).$$
In the first sum we have $\forall u\not\in \ell\leftrightarrow \ell': x_u'=x_u''$ because the configuration value at these nodes is not affected by moves of the form $\ell\rightarrow\ell'$. 
In the second sum, we observe that by optimality of the transport map, the sign of $x_u'-x_u$ is always the same as the sign of $z_u$. This is because in an optimal transport plan, no two robots are transported on the same edge in opposite directions. Therefore we have the following inclusion on segments $[x_u+z_u,x_u'-z_u]\subset [x_u,x_u']$, which by convexity of $\phi$ implies that the chord of the inner segment is below the chord of the outer segment, $\phi(x_u+z_u)+\phi(x_u'-z_u)\leq \phi(x_u)+\phi(x_u')$. Rearranging the terms, we get the following identity, $\phi(x_u)-\phi(x_u') \leq \phi(x_u)-\phi(x_u'')+\phi(x_u)-\phi(x_u+z_u)$. Combining these observations yields, 
    $$\Psi(\bx)-\Psi(\bx') \leq \sum_{u\in V}d_u(\phi(x_u)-\phi(x_u''))+ \sum_{u\in \ell\leftrightarrow \ell'}d_u(\phi(x_u)-\phi(x_u+z_u)),$$
    where we recognize the first term to be $\tau(\bx\rightarrow \bx'')$ and the second term to be $\tau_\bx(\ell\rightarrow \ell')$. We conclude by applying the induction hypothesis. 

    By strict convexity of $\phi$, it is a direct observation that the inequality above is strict as soon as the inclusion $[x_u+z_u,x_u'-z_u]\subset [z_u,z_u']$ is strict, i.e. when there are overlaps in the transport plan. 
\end{proof}
The rest of the proof of Proposition \ref{prop_dynamics_x} is decomposed into \cref{lemma:elongation-tension,lemma:deletion-tension,lemma:fork-tension,lemma:minoring-tension}, in Appendix \ref{ap:control_of_x}. These lemmas each treat separately the cases of leaf elongation, leaf deletion and leaf fork, all of them relying on Lemma \ref{lemma:tension-decomposition}.

\subsection{Competitive analysis by potential function}\label{sec:compet-analysis}

For a given tree $T$, we consider the optimal \textit{fractional} configuration $\by=\{y_\ell\}_{\ell\in \cL(T)} \in \Ycal(T)$ defined as the solution of the following convex optimization problem,
\begin{equation}\label{eq:opt-problem}
\begin{array}{ll}
\min &\Psi(T,\by)=\sum_{u\in V\setminus \{r\}}d_u \phi(y_u)\\
\hbox{over}&\by\in \Ycal(T),
\end{array}
\end{equation}
where we further assume that $\phi$ is twice differentiable, with $\phi''>0$ and $\phi'> 0$.

The goal of this section is for some suitable constant $\gamma$, and some choice of function $\phi$, to prove that the following equation is verified at all times, irrespective of the actions taken by the adversary,
\begin{equation}\label{eq:master_inequality}
\Cost(t)+\Psi(T(t),\bx(t))\le \gamma \Psi(T(t),\by(t)).
\end{equation}
This readily leads to conclude that the strategy is $f(k,D)$-bounded, with $f(k,D) = \gamma\phi(k)D$. This is because $\gamma \phi(k)D$ is a simple upper bound of \eqref{eq:opt-problem}. The rest of this section is dedicated to sowing \eqref{eq:master_inequality} for a quadratic $\phi(\cdot)$ and a constant $\gamma$, thereby proving Theorem \ref{th:main}. We thus analyze the behaviour of $\by$, starting with Proposition~\ref{proposition: static} which provides a characterization of $\by$. We then study the dynamics of $\by$ over the course of the game (elongations, deletions, forks) in Proposition~\ref{proposition:dynamics}. After, we prove some relations between $\bx$ and $\by$, for some quadratic function $\phi$ in Proposition \ref{prop:1}. Finally, we conclude with the proof of \eqref{eq:master_inequality} in Proposition~\ref{prop:potential-arguement}. 

\begin{proposition}[Characterization of $\by$]\label{proposition: static}
For any tree $T$ of the game, Equation \eqref{eq:opt-problem} has a unique solution. Furthermore, $\by\in \Ycal(T)$ is the solution of \eqref{eq:opt-problem} if and only if there exists $\lambda\in \Rbb$ s.t.
\begin{equation}\label{eq:lagrange}
 \forall \ell\in\cL(T), \sum_{u\in \ell \rightarrow r}d_{u}\phi'(y_u)\ge \lambda, \quad \text{and} \quad y_\ell>0\Rightarrow \sum_{u\in \ell \rightarrow r}d_{u}\phi'(y_u) = \lambda.   
\end{equation}
Consequently, given any two leaves, $\ell,\ell'$ for which $y_\ell,y_\ell'>0$ we have
\begin{equation}\label{eq:equilibrium}
    \sum_{u\in \ell \rightarrow \LCA(\ell,\ell')}d_u\phi'(y_u) = \sum_{u\in \ell' \rightarrow \LCA(\ell,\ell')}d_u\phi'(y_u).
\end{equation}

\end{proposition}
\begin{proof} We start by justifying the existence and uniqueness of $\by$. We note that for any tree $T$, the function $\by\rightarrow \Psi(T,\by)$ is twice differentiable and compute its Hessian,
\begin{equation}\label{eq:hessian}
    \nabla^2_\by \Psi(T,\by) = \left(
    \sum_{u\succeq \LCA(\ell,\ell')}d_u\phi''(y_u)
\right)_{\ell,\ell' \in \Lcal(T)}
\end{equation}
which is an ultrametric matrix (see \cite[Definition 3.2, p. 58]{dellacherie2014inverse}), because $\phi''> 0$. We recall that an ultrametric matrix is always semidefinite positive because it can be seen as the covariance matrix of a Brownian motion on a tree (see e.g. \cite{sturmfels2019brownian}). Furthermore, it is non-singular if no two rows are equal \cite[Th 3.5 (ii)]{dellacherie2014inverse}. This applies here since $\forall u\in V: d_u>0$. Thus $\by \rightarrow\Psi(T,\by)$ is strictly convex and its minimum on the compact and convex set $\Ycal(T)$ exists and is unique.

We now give a characterization of $\by$ by applying the KKT conditions. Let $\mu=(\mu_\ell)_{\ell\in \cL(T)}$ denote the vector of non-negative Kuhn-Tucker multipliers associated with inequality constraints $y_\ell \geq  0$, and $\lambda\in \Rbb$ the multiplier associated with equality constraint $\sum_{\ell\in\cL(T)}y_\ell=k$. We can characterize $\by$ by forming the Lagrangian 
$$
L(\bz;(\lambda,\mu)):=\sum_{u\in V(T)\setminus\{r\}}d_u\phi(z_u)-\sum_{\ell\in\cL(T)}\mu_\ell z_\ell+\lambda\left[k-\sum_{\ell\in \cL(T)}z_\ell\right].
$$
Optimality of $\by$ is  characterized by the existence multipliers $\lambda$, $\mu$ such that the stationarity conditions together with complementary conditions are satisfied, 
$$
\forall \ell\in \cL(T), \sum_{u\in \ell \rightarrow r}d_{u}\phi'(y_u)-\mu_\ell-\lambda=0,
\quad\text{and}
\quad
\forall \ell\in\cL(T),\; \mu_\ell y_\ell=0.
$$
This is readily seen to be equivalent to the conditions stated in the Lemma.
\end{proof}

\begin{proposition}[Dynamics of $\by(t)$ following \eqref{eq:opt-problem}]\label{proposition:dynamics}
Consider some time interval $I=[t_1,t_2]$, in which all leaves $\ell$ satisfy $y_\ell>0$ and for which the only edge length to increase is $d_l$, then $y_l$ decreases and for any leaf $\ell\neq l$, $y_\ell$ increases. Furthermore, $\by$ is derivable and its derivative $\dot{\by}$ satisfies
\begin{equation}\label{eq:controlled_decay}
\dot{y}_l\ge - \frac{1}{d_l}\frac{\max \phi'}{\min\phi''}.
\end{equation}
Consider a leaf $l$ that undergoes a discrete step (fork or deletion), then the optimal configuration $y_\ell$ of any leaf $\ell\neq l$ is increased.  
Furthermore, if $l$ undergoes a fork of length $\delta$ with $m$ children, one has 
\begin{equation}\label{eq:controlled_fok}
y_l'-y_l\geq -\frac{\delta}{d_l}\frac{\max \phi'}{\min\phi''},
\end{equation}
where $\by'$ denotes the optimal configuration right after the fork took place, in the tree in which $l$ is now the parent of $m$ children at distance $\delta$ and all attributed configuration weight $y_l'/m$. 
\end{proposition}
\begin{remark} The assumption that $y_\ell>0$ is not required in practice, but it simplifies the argument and it is always verified for the given Algorithm \eqref{eq:dynamics} of the player, as we shall later see.
\end{remark}
\begin{proof} We treat separately the cases of elongation, deletion and fork. 

\paragraph{Leaf elongation.}  Consider the leaf $l$ which is being elongated during interval $[t_1,t_2]$, all other edges being left unchanged. Consider some $t\in [t_1,t_2)$ 
 satisfying $\by(t)>0$ and some small $dt>0$. The derivability of $\by(\cdot)$ at $t$ can be deduced from straightforward arguments using the inverse function theorem and the smooth evolution of the convex potential. We thus focus here on the sign of derivatives $\dot{\by}(t)$. We have by the characterization of $\by(t+dt)$ and of $\by(t)$ given in \eqref{eq:lagrange} that $\nabla_\by \Psi(T(t+dt),\by(t+dt)) - \nabla_\by \Psi(T(t),\by(t))\in \text{Vect}(\oneb)$. Using the identity $\nabla_\by \Psi(T(t+dt),\by(t+dt)) = \nabla_\by \Psi(T(t),\by(t+dt))+\phi'(y_l(t+dt))\be_l dt$ and a Taylor expansion, we obtain, 
$$\nabla^2_{\by} \Psi(T(t),\by(t))(\by(t+dt)-\by(t))+\phi'(y_l(t+dt))\be_ldt + o(dt)\in \text{Vect}(\oneb),$$
which in the limit $dt\rightarrow 0$ gives, 
\begin{equation}\label{eq:differential}
    U(t) \dot{\by}(t) \in -\phi'(y_l(t))\be_l+\text{Vect}(\one)
\end{equation}
where for shorthand we denote $U(t) = \nabla^2_{\by} \Psi(T(t),\by(t))$, the ultra-metric defined in \eqref{eq:hessian}.
Lemma \ref{lemma:ultrametric_equation} below applied to \eqref{eq:differential} then allows to conclude that $\dot{y}_l<0$ and $\dot{y}_\ell\geq 0$ for all $\ell\neq l$.  
\begin{lemma}\label{lemma:ultrametric_equation}
    Assume $U\in \Rbb^{n\times n}$ a positive definite ultrametric matrix satisfies
    \begin{equation}\label{eq:ultrametric-eq}
        U \bz = -\mu\be_l+\lambda \oneb
    \end{equation}
for constants $\lambda \in \Rbb$, $\mu>0$, for some index $l\in [n]$ and some zero sum vector $\bz \in \Rbb^{n}$, i.e. such that $\bz^T \oneb=0$. It is then the case that $\lambda\geq 0$ and that $\forall \ell\in [n]\setminus\{l\}: z_\ell\geq 0$.
\end{lemma}
\begin{proof}We denote by $M$ the inverse of the positive-definite ultrametric $U$, i.e. $M=U^{-1}$, which satisfies the following properties, \cite[Th. 3.5 (i)]{dellacherie2014inverse} (a) the diagonal elements of $M$ are non-negative (b) the off-diagonal elements of $M$ are non-positive (c) the sum of elements of $M$ across a row (or column) is non-negative. 

By multiplying \eqref{eq:ultrametric-eq} on the left by $\oneb^T M$ and using that $\one^T\bz =0$, we get $\mu\one^T M \be_l = \lambda \one^T M \one$. 
By property (c) above,  $\one^T M \be_l\geq 0$ and since $\mu\geq 0$ we get $\lambda\geq 0$. 
We then have $z = -\mu M e_l+\lambda M\one$, thus for $\ell \neq l$ we get $z_\ell = -\mu\be_\ell^T M \be_l + \lambda \be_\ell^T M \one$. Since $\be_\ell^T M \be_l \leq 0$ by property (b) and $\be_\ell^T M \one\geq 0$ by property (c) we get $z_\ell\geq 0$. Finally, observe that since $\one^T\bz =0$, it must be the case that $z_l \leq 0$. 
\end{proof}
We now consider a leaf $\ell\ne l$ that is a descendant of $p(l)$. Using equation \eqref{eq:equilibrium} we get $\sum_{u\in \ell \rightarrow p(l)}d_u\phi'(y_u) = d_l\phi'(y_l)$ at all times. Taking time derivatives, we obtain, 
$
\sum_{u\in \ell\rightarrow p(l)}d_{u}\dot{y}_{u} \phi''(y_u)=\phi'(y_l)+d_l\dot{y}_l \phi''(y_l).
$
We have shown that the left-hand side is non-negative, thus ensuring,
$$
\dot{y}_l\ge - \frac{1}{d_l}\frac{\phi'(y_l)}{\phi''(y_l)}\geq - \frac{1}{d_l}\frac{\max \phi'}{\min\phi''}.
$$
\paragraph{Leaf fork.} We now consider the case when leaf $l$ undergoes a discrete fork with $m$ children, each attached to an edge of length $\delta$. 
This case will be treated similarly as that of a continuous elongation by introducing $s\in [0,1]$ and $T(s)$ the tree constructed from $T$ by providing $l$ with $m$ children each attached to a leaf of length $s\delta$. It is clear that $T(0)$ corresponds to the tree before the fork and that $T(1)$ corresponds to the tree after the fork. We shall show that optimal configuration $\by(s)$ of $T(s)$ satisfies the property that $y_\ell(s)$ is increasing if $\ell$ is not a child of $l$. 
It is clear from the strict convexity of $\by\rightarrow\Psi(T(s), \by)$ and by symmetry that the value of $y_\ell(s)$ for all children $\ell$ of $l$ is equal to $\frac{1}{m}y_l(s)$, as soon as $s>0$. We will thus see $\by(s)$ as a configuration of the tree before the fork, but for the modified potential $\Psi(s,\by) = \Psi(T(0),\by) + s\delta\times m \phi\left(\frac{1}{m}y_l\right)$. By the same reasoning as above, we get the following dynamics 
$$U(s) \dot{\by}(s)\in -\delta\phi'\left(y_l(s)/m\right)\be_l+\text{Vect}(\oneb),$$
where $U(s) = \left(
    \sum_{u\succeq \LCA(\ell,\ell')}d_u\phi''(y_u(s))
\right)_{\ell,\ell' \in \Lcal(T)}+\left(\one(\ell=\ell'=l)\frac{s \delta}{m}\phi''(y_l(s))\right)_{\ell,\ell' \in \Lcal(T)}$ is a positive definite ultra-metric to which we apply Lemma \ref{lemma:ultrametric_equation} to get the desired monotonicity. Then by considering a leaf $\ell\neq l$ that is a descendant of $p(l)$, we have $\sum_{u\in \ell \rightarrow p(l)}d_u\phi'(y_u(s)) = d_l\phi'(y_l(s))+\delta s\phi'\left(y_l(s)/m\right)$. By taking derivatives on $s$, the quantity, 
$d_l\phi''(y_l(s))\dot{y}_l(s) +\frac{s\delta}{m}\phi''\left(y_l(s)/m\right)\dot{y}_l(s)+\delta\phi'\left(y_l(s)/m\right)$
is non-negative, ensuring that 
$\dot{y}_l(s)\geq -\frac{\delta}{d_l}\frac{\phi'\left(y_l(s)/m\right)}{\phi''(y_l(s))},$
and thus
$y_l'-y_l\geq -\frac{\delta}{d_l}\frac{\max \phi'}{\min\phi''}.$

\paragraph{Leaf deletion.} Finally, we consider the case of the deletion of leaf $l$. We observe as in \cite{bubeck2022shortest}, that the configuration after the deletion of a leaf can be obtained as the limit configuration when this leaf is extended to infinity. Such infinite extension only leads to increasing the value of the configuration at all other leaves, thereby proving the monotonicity statement made for deletions.\end{proof} 

\begin{proposition}[Bounds on $\bx$ and $\by$]\label{prop:1}
Let $\epsilon, \epsilon' \in (0,1/2]$ be two fixed constants, and assume that $\phi(x) = ax+bx^2$ where parameters $a,b$ are chosen such that
\begin{equation}\label{eq:assumptions}
2b k \le \epsilon' a\hbox{ and }b(2-2\epsilon-\epsilon')\ge 2+\epsilon'.
\end{equation}
Then Algorithm \eqref{eq:dynamics} is such that irrespective of the adversary's moves, at all times, $\bx\in \Xcal(T)$ verifies 
\begin{equation}\label{eq:epsilon}
\forall \ell\in\cL(T),\; x_\ell < y_\ell+2-\epsilon.
\end{equation}
Also, in light of Proposition \ref{proposition:dynamics}, we have that at all times and irrespective of the adversary's moves,
\begin{equation}\label{eq:lowerbound}
\forall \ell\in\cL(T),\; \epsilon \leq y_\ell.
\end{equation}
\end{proposition}
\begin{proof}We start by showing how \eqref{eq:lowerbound} follows from \eqref{eq:epsilon}. Recall from Proposition \ref{proposition:dynamics} that the value of $y_\ell$ can only decrease if $\ell$ is elongated or forked. Since a leaf can only be elongated if $x_\ell\geq 2$, it follows from \eqref{eq:epsilon} that $y_\ell$ can not go below $\epsilon$ during an elongation. Now assume that the leaf $\ell$ undergoes a fork with $m\leq x_\ell-1$ children. By \eqref{eq:epsilon}, it must be the case that $y_l\geq m-1+\epsilon\geq m\epsilon + (1-\epsilon)$. For a small enough value of $\delta$ in equation \eqref{eq:controlled_fok}, we get that $y_l'/m>\epsilon$, just after the fork. Thus neither elongation nor a fork may make a value of $\by$ go beneath $\epsilon$. 

We now show how \eqref{eq:epsilon} follows from \eqref{eq:assumptions}. Note first that \eqref{eq:epsilon} holds true at $t=0$. Consider then the earliest time at which this inequality is met with equality: there is some leaf $\ell$ such that $x_\ell=y_\ell+2-\epsilon$. Then, since $x_r=y_r=k$, there exists an ancestor $v$ of $\ell$ such that $x_v<y_v+2-\epsilon$ but $x_u\geq y_u+2-\epsilon$ for all $u\in \ell\rightarrow v$. Also, since $x_v =\sum_{u: p(u)=v}x_u$ and $y_v =\sum_{u: p(u)=v}y_u$ there must be a child $w$ of $v$ such that $x_w<y_w$. By iterating this argument until we reach a leaf, there must exist $\ell' \in \Lcal(T)$ descending from $v$ and satisfying that for all $u\in \ell'\rightarrow v : x_u<y_u$. Let $\bx'$ be the robot configuration obtained from $\bx$ by moving one robot from $\ell$ to $\ell'$. We evaluate the difference
$$
\Delta:= \Psi(T,\bx')+\OT_T(\bx,\bx')-\Psi(T,\bx).
$$
This reads
\begin{equation*}
\begin{array}{ll}
\Delta&= \displaystyle \sum_{u\in \ell\rightarrow v} d_{i(j)}[1-a +b\{(x_{u}-1)^2-x_{u}^2\}]+\sum_{u\in \ell'\rightarrow v}d_{u}[1+a+b\{(x_{u}+1)^2-x_{u}^2\}]\\
&=\displaystyle
\sum_{u\in \ell\rightarrow v} d_{u}[1-a +b(1 -2 x_{u})]+\sum_{\ell'\rightarrow v}d_{u}[1+a+b( 1 +2 x_{u})]\\
&\le
\sum_{u\in \ell\rightarrow v} d_{i(j)}[1-a +b(1 -2 (y_u+2-\epsilon))]+\sum_{\ell'\rightarrow v}d_{i'(j)}[1+a+b( 1 +2 y_u)],
\end{array}
\end{equation*}
where we used the inequalities between $\bx$ and $\by$. 

We then recall from \eqref{eq:equilibrium} (noting that $\phi'(x)=a+2bx$):
\begin{equation}\label{eq:identité_xstar}
\sum_{u\in \ell\rightarrow v}d_{u}[a+2b y_u]=\sum_{u\in \ell'\rightarrow v}d_{u}[a+2b y_u].
\end{equation}
We can thus simplify the previous bound on $\Delta$ to obtain
\begin{equation}\label{eq:delta2}
\Delta\le \sum_{u\in \ell\rightarrow v} d_{u}[1+b(1 -2 (2-\epsilon))]+\sum_{u\in \ell'\rightarrow v}d_{u}[1+b].
\end{equation}
Equation \eqref{eq:identité_xstar} together with the assumption $2b k\le \epsilon' a$ made in the proposition entails, noticing that $y_v\le k$ for all $v$, that
\begin{equation}\label{eq:equilibre-branches}
\sum_{u\in \ell\rightarrow v}d_{u}\le (1+\epsilon')\sum_{u\in \ell\rightarrow v} d_{u}, 
\quad \text{and similarly,} \quad 
\sum_{u\in \ell'\rightarrow v} d_{u}\le (1+\epsilon')\sum_{u\in \ell\rightarrow v}d_{u}.
\end{equation}
In view of \eqref{eq:delta2}, we thus have
$$
\begin{array}{ll}
\Delta&\le \displaystyle \sum_{u\in \ell\rightarrow v}d_{u}\left[1+b(-3+2\epsilon)+(1+\epsilon')(1+b)\right]= \sum_{u\in \ell\rightarrow v}d_{u}\left[2+\epsilon' -b(2-2\epsilon-\epsilon')\right].
\end{array}
$$
This upper bound is negative under the second condition $b(2-2\epsilon-\epsilon')\ge 2+\epsilon'$ of the proposition. Thus by \eqref{eq:dynamics} some move by the player must take place if for some $\ell$ one has $x_\ell=y_\ell +2-\epsilon$.
\end{proof}

\begin{remark}\label{rem:param}
We will will choose as specific values for $\phi(x) = ax+bx^2$ the quantities $a=20 k$, and $b=5$, as well as $\epsilon=\epsilon'=1/2$. These quantities satisfy the conditions of Proposition \ref{prop:1}. 
\end{remark}

\begin{proposition}\label{prop:potential-arguement}
The above-defined strategy, for the choice of parameters made in Remark \ref{rem:param}, is such that at all times $t$, the configuration $(T(t),\bx(t))$ of the game satisfies \eqref{eq:master_inequality} for $\gamma=48$. 
\begin{equation*}
\Cost(t)+\Psi(T(t),\bx(t))\le \gamma \Psi(T(t),\by(t)). \tag{\ref{eq:master_inequality}}
\end{equation*}
\end{proposition}

\begin{proof}
We show \eqref{eq:master_inequality} is preserved by leaf elongations, leaf deletions and leaf forks. 

\paragraph{Leaf elongation.} We first consider the evolution of the equation during a leaf elongation and between instants at which the player re-assigns robots. The three quantities involved in \eqref{eq:master_inequality} evolve smoothly. Denoting by $\ell(t)$ the leaf chosen by the adversary for elongation, we have:
$$
\begin{array}{l}
\frac{d}{dt}\Cost(t)=x_{\ell(t)}(t),\\
\frac{d}{dt}\Psi(T(t),\bx(t))=\phi(x_{\ell(t)}(t)),\\
\frac{d}{dt}\Psi(T(t),\by(t))=\phi(y_{\ell(t)}(t)),
\end{array}
$$
where for the third identity we used the optimality of $\by(t)$.

We thus want to choose $\gamma$ such that, writing $x=x_{\ell(t)}(t)$ and $y=y_{\ell(t)}(t)$ for short, the following inequality holds:
\begin{equation}\label{eq:competitive}
    x+\phi(x)\le \gamma \phi(y).
\end{equation}
Proposition \ref{prop:1} gives us the following relations:
$$
x\le y+2-\epsilon=y+3/2 \quad \text{and} \quad y\ge \epsilon=1/2.
$$
We then notice that $\forall y\geq 1/2: y+3/2+\phi(y+3/2)\leq \gamma \phi(y)$ is satisfied for any $\gamma \geq 7$, and for the choice of parameters made in Remark \ref{rem:param}, therefore proving \eqref{eq:competitive}.

Next, we consider an instant $t$, occurring during a leaf elongation, at which the player reassigns robots, moving from $\bx$ to $\bx'$. The cost incurred by the player for this move is precisely $\OT_{T(t)}(\bx,\bx')$. Because of the player's strategy, this move occurs only if
$$
\OT_{T(t)}(\bx,\bx')+\Psi(T(t),\bx')\le \Psi(T(t),\bx).
$$
Thus the left-hand side of \eqref{eq:master_inequality} makes a downward jump at time $t$, whereas its right-hand side incurs no jump.

\paragraph{Leaf fork.} The case of leaf fork is similar to the case of leaf elongation. Denote by $T$, $\bx$ and $\by$ (resp $T'$, $\bx'$ and $\by'$) the tree, the discrete configuration and the optimal configurations before (resp. after) the fork and denote by $\ell$ the leaf that is forked. Recall that the fork length $\delta$ is chosen small enough for the fork to lead to no reassignments, i.e. $x_\ell=x_\ell'$ and the newly created leaves take their values in $\{\ceil{x_\ell/m},\floor{x_\ell/m}\}$. The fork leads to the following change in cost and potential,
$$
\begin{array}{l}
\Delta_\bx\Cost=\delta x_{\ell},\\
\Delta_\bx\Psi=\Psi(T',\bx')-\Psi(T,\bx)\leq m \delta \phi(\ceil{x_\ell/m}),\\
\Delta_\by\Psi=\Psi(T',\by')-\Psi(T,\by)\geq \int_0^\delta m\phi(y_\ell(s)/m)ds \geq  m\delta\phi(y_\ell'/m),
\end{array}
$$
where for the third inequality, we used the fact that a leaf fork can be seen as a leaf elongation with modified leaf potential $:y\rightarrow m\phi(y/m)$, see Proposition \ref{proposition:dynamics}. Further by Proposition \ref{proposition:dynamics} the value of $\delta>0$ can be chosen small enough to ensure $y_\ell'/m\geq \epsilon$ and $y_\ell'\geq y_\ell-1/2$. We then use $x_\ell \leq y_\ell+3/2\leq y_\ell'+2$ to observe that 
$\Delta_\bx \Cost \leq \Delta_\by \Psi,$
and use $\ceil{x_\ell/m} \leq y_\ell'/m+3/2$ and $y_\ell'/m\geq \epsilon$ along with the computations used for leaf elongations to get that 
$\Delta_\bx \Psi \leq \gamma \Delta_\by \Psi'$
for any $\gamma \geq 7$. Summing up both inequalities, we get that for any $\gamma \geq 8$,
$$\Delta_\bx \Cost+\Delta_\bx \Psi \leq \gamma \Delta_\by \Psi'.$$

\paragraph{Leaf deletion.} We finally deal with leaf deletions and keep the same notations of $(T,\bx,\by)$ and $(T',\bx',\by')$ to denote the state of the game before and after the deletion and use $\ell$ to denote the deleted leaf. We first lower-bound the value of $\Delta_\by \Psi = \Psi(T',\by')-\Psi(T,\by)$ and then upper bound the value of $\Delta_\bx (\Psi+\Cost) = \Psi(T',\bx')-\Psi(T,\bx)+\OT(\bx,\bx')$.

We define $\bdelta = \by'-\by \in \Ycal(T)$ where we set by convention $y_\ell'=0$. For each $s\in[0,1]$, we then let $\by(s) = \by+s\bdelta\in \Ycal(T)$. We then define $h(s):=\sum_u d_u \phi(y_u(s))$, so that 
$$
h(0)=\Psi(T,\by),\; h(1)=\Psi(T',\by').
$$
By optimality of $\by$ for $T$, $h'(0)=0$. Also,
$$
h''(s)=\sum_u d_u \delta_u^2 (2b)\ge d_\ell (y_\ell)^2(2b).
$$
Thus $h'(s)\ge 2bs * d_\ell(y_\ell)^2$. In turn one obtains
\begin{equation}\label{eq:tmp_rhs_increase}
\Delta_\by \Psi = \Psi(T',\by')-\Psi(T,\by)=h(1)-h(0)\ge bd_\ell(y_\ell)^2.
\end{equation}
Let us now upper-bound the increase to the left-hand side of \eqref{eq:master_inequality} incurred by the player's move. Note that the player's strategy is to choose a new configuration $\bx'\in\Ycal(T')$ that minimizes $\OT_{T}(\bx,\bx')+\Psi(\bx')-\Psi(\bx)$, which is precisely the amount by which this left-hand side will increase. We can thus choose any target configuration $\bx'' \in \Ycal(T')$ we like to upper-bound this left-hand side increase. Pick one leaf $\ell'\neq \ell$ that is a descendant of $p(\ell)$, and consider the configuration $\bx'' = \bx - x_\ell\be_{\ell}+x_{\ell}\be_{\ell'}$ obtained by moving the $x_\ell$ robots from $\ell$ to $\ell'$. Using \eqref{eq:equilibre-branches}, we observe that the associated transport cost $\OT(\bx,\bx'')$ is no larger than $(5/2)d_\ell x_\ell$. The difference in potential then writes as follows
$$\Psi(\bx'')-  \Psi(\bx) = \sum_{u \in \ell\rightarrow \LCA(\ell,\ell')}d_u(\phi(x_u-x_\ell)-\phi(x_u))+\sum_{u\in \ell'\rightarrow \LCA(\ell,\ell')}d_u(\phi(x_u+x_\ell)-\phi(x_u)).$$
We then observe that for $\phi(x) = ax +bx^2$, we have for $x,h \geq 0$, 
\begin{align*}
    \phi(x+h)-\phi(x) &=  bh^2 + (a+2bx)h \leq bh^2 + (\phi(x+1)-\phi(x))h\\
    \phi(x-h)-\phi(x) &= bh^2-(a+2bx)h\leq bh^2+(\phi(x-1)-\phi(x))h
\end{align*}
and thus, 
\begin{align*}
\Psi(\bx'')-  \Psi(\bx) &= bx_\ell^2d(\ell,\ell') +x_\ell \left(\sum_{u \in \ell\rightarrow \LCA(\ell,\ell')}d_u(\phi(x_u-1)-\phi(x_u))+\sum_{u\in \ell'\rightarrow \LCA(\ell,\ell')}d_u(\phi(x_u+1)-\phi(x_u))\right)\\
&\leq bx_\ell^2d(\ell,\ell') +x_\ell \tau_\bx(\ell'\rightarrow\ell)\\
&\leq bx_\ell^2d(\ell,\ell') +x_\ell d(\ell,\ell')
\end{align*}
where in the last equation we used the inequality $\tau_x(\ell \rightarrow \ell')\leq d(\ell,\ell')$, which was proved in Proposition \ref{prop_dynamics_x}. 
This yields the following bound,
$$\Psi(\bx'')-  \Psi(\bx)
\leq bx_\ell^2d(\ell,\ell') +x_\ell d(\ell,\ell')
\leq d(\ell,\ell')(y_\ell)^2(16b+8),$$
where we used $x_\ell \leq 4y_\ell$ which follows from \eqref{eq:epsilon} and $y_\ell\geq 1/2$. 
Adding the transport cost of $\OT(\bx'',\bx) = x_\ell d(\ell,\ell')$ we get, 
$$\Delta_\bx (\Psi+\Cost) \leq (y_\ell)^2d(\ell,\ell')(16b+16)\leq (y_\ell)^2d_\ell(5/2)(16b+16) \leq \gamma \Delta_\by \Psi$$
for $\gamma = (5/2)(16b+16)/b = 48$ where we used the lower bound of $\Delta_\by \Psi$ given by \eqref{eq:tmp_rhs_increase}.
\end{proof}

\section{Reductions of Continuous Tree-Mining Game}\label{sec:reductions}

In this section, we precise the reductions that connect the game above to both of our motivations, collective tree exploration (CTE) and layered tree traversal (LTT). The reductions are summarized in the following diagram, Figure~\ref{fig:diagram-reductions}. 
\begin{figure}[!h]
\centering
\begin{tikzpicture}[node distance=3cm, every node/.style={draw, align=center}]

\node (box1) {Continuous Tree-Mining\\ Game (CTM)};

\node[below of=box1] (box2) {Discrete Tree-Mining\\ Game (TM)};

\node[right of=box2, xshift=3cm, rounded corners=0.3cm] (box3) {Asynchronous Collective \\ Tree Exploration (ACTE)};

\node[below of=box3, xshift=-3cm, rounded corners=0.3cm] (box4) {Collective \\ Tree Exploration (CTE)};

\node[below of=box3, xshift=3cm, rounded corners=0.3cm, double] (box5) {Layered Tree\\ Traversal (LTT)};

%
\draw[->] (box1) -- 
node[midway, right,font=\small,draw=none] {Prop. \ref{th:tm-ctm}} 
(box2);
\draw[<->] (box2) -- 
node[midway, above,font=\small,draw=none] {\cite{cosson2023breaking}} 
node[midway, below,font=\small,draw=none] {Prop. \ref{th:tm-acte}} 
(box3);
\draw[->] (box3) -- 
node[midway, left,font=\small,draw=none] {\cite{cosson2023breaking} \& Prop. \ref{th:rcte}} 
(box4);

\draw[->] (box3) -- 
node[midway, right,font=\small,draw=none] {Prop. \ref{th:reduction-lgt}} 
(box5);

\end{tikzpicture}

\caption{Summary of reductions. Rectangles are for settings expressed in the form of a game. Ovals are for exploration problems. Double boundary are for online problems.}
\label{fig:diagram-reductions}
\end{figure}
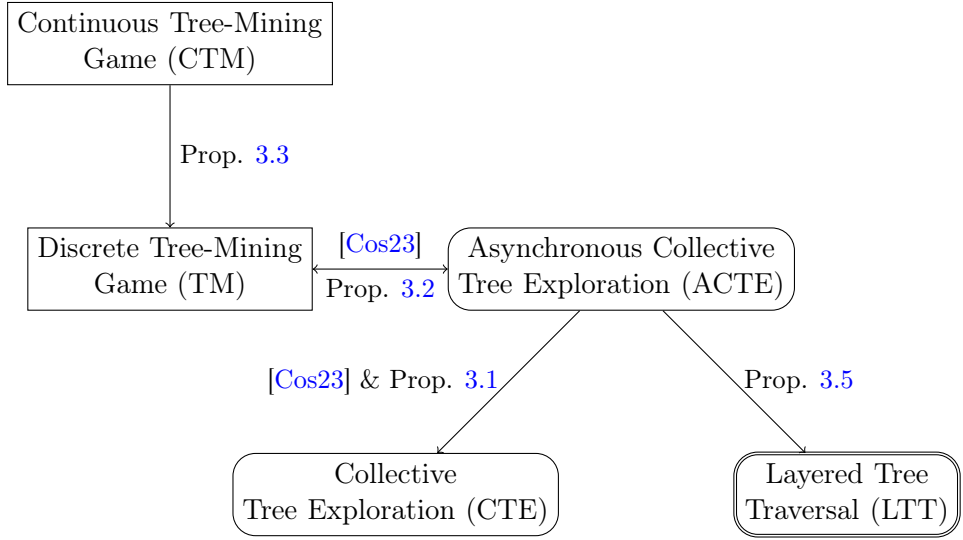


\subsection{Synchronous and Asynchronous Collective Tree Exploration}
The problem of (synchronous) collective tree exploration (CTE) was introduced by \cite{fraigniaud2006collective}. We shall first recall the definition of the problem and then present an extension called asynchronous collective tree exploration (ACTE) \cite{cosson2023breaking}. We finally recall the reduction from (CTE) to (ACTE) in Proposition \ref{th:rcte}. 

\paragraph{(Synchronous) Collective Tree Exploration (CTE).}
The problem is defined as follows. A team of $k$ robots is initially located at the root of an unknown tree $T$ with edges of length $1$. The team is tasked to go through the edges of a tree as fast as possible and then return to the root. At all rounds, the robots can move synchronously along one incident edge, thereby possibly visiting new nodes. When a robot visits a node $u$ for the first time, the full team becomes aware of its degree and of all the unexplored edges adjacent to $u$, that we shall call \textit{dangling} edges. When there are no more dangling edges, the entire team knows that all edges have been explored and meets at the root to declare the termination of exploration. The runtime of an exploration algorithm $\Acal$ is defined on a tree $T$ as the number of rounds before termination and is denoted by $\texttt{Runtime}(\Acal,k,T)$. By extension, for any two integers $n\geq D$, we denote by $\texttt{Runtime}(\Acal,k,n,D)$ the maximum number of rounds required before termination on any tree with $n$ nodes and depth $D$. For more formalism on the definition of synchronous collective tree exploration with complete communication, we refer to \cite[Section 2]{cosson2023breadth}. 

\paragraph{Asynchronous Collective Tree Exploration (ACTE).} The problem of asynchronous collective tree exploration (ACTE) is a generalization of collective tree exploration (CTE) in which the agents move sequentially. At each (discrete) time $t$, only a single robot indexed by $r_t\in [k]$ is allowed to move. The sequence $r_0,r_1,\dots$ is completely adversarial. At the beginning of move $t$, the team is \textit{only} given the information of whether robot $r_t$ is adjacent to some unexplored edge, and if so, the ability to move along that edge. 
Consequently in contrast with collaborative tree exploration, the team does not know the degree of a node until it has been \textit{mined}, i.e. some robot moving from that node was not given a fresh unexplored edge to traverse. Exploration starts with all robots located at the root of some unknown tree $T$ and ends only when all nodes have been mined. We do not ask that all robots return to the root at the end of exploration for this might not be possible in finite time if say some robot is allowed only a single move. 
For an asynchronous collective tree exploration algorithm $\Bcal$, we denote by $\texttt{Moves}(\Bcal,k,T)$ the maximum number of moves that are needed to traverse all edges of $T$. We naturally extend this notation to $\texttt{Moves}(\Bcal,k,n,D)$ for two integers $n\geq D$ denoting the maximum number of moves required to explore a tree with $n$ nodes and depth $D$. (ACTE) generalises (CTE) in the following sense.
\begin{proposition}
    \label{th:rcte}
    For any asynchronous exploration collaborative tree exploration algorithm $\Bcal$, one can derive a synchronous algorithm $\Acal$ satisfying,
    \begin{equation*}
        {\normalfont\texttt{Runtime}}(\Acal,k,T) \leq \left\lceil\frac{1}{k}{\normalfont\texttt{Moves}}(\Bcal,k,T)\right\rceil+D,
    \end{equation*}
    for any tree $T$ of depth $D$, and for all number of robots $k$. 
\end{proposition}
\paragraph{Proof idea.} Considering a sequence of allowed moves $1,2, \dots, k,1,2, \dots, k, 1,2, \dots$, we recover a synchronous collective exploration algorithm from an asynchronous collective exploration algorithm. Then all robots return to the root in at most $D$ additional steps. We refer to \cite{cosson2023breaking} for more details.

\subsection{Continuous and Discrete Tree-Mining Games}
The problem of asynchronous collective tree exploration can be reduced to a game opposing a player and an adversary known as the `tree-mining' game (TM) \cite{cosson2023breaking}.
In this section, we first recall the original setting of the discrete tree-mining game (TM), and then show it can be reduced to
the continuous tree-mining game (CTM), which was studied in Section \ref{sec:analysis}.

\paragraph{Discrete Tree-Mining Game (TM).} The game is defined as follows for some fixed integer $k\geq 2$. At step $i\geq 0$, the state of the game is defined by a pair $\Tcal(i) = (T(i), \bx(i))$ where $T(i)=(V(i),E(i),L(i))$ is a rooted tree with nodes $V(i)$, edges $E(i)$, and with $L(i)\subset \Lcal(T(i))$ a subset of the leaves of $T(i)$ called \textit{active} leaves. The configuration $\bx(i) = (x_\ell(i))_{\ell\in L(i)}$ represents the number of miners on each active leaf, for a total of $k$ miners, i.e. $\sum_{\ell\in L(i)} x_{\ell}(i)=k$. At the beginning of step $i\in \Nbb$, the adversary is the first to play. 
It chooses an active leaf $\ell(i)\in L(i)$ that becomes inactive and is given $x_{\ell(i)}(i)-1$ (active) children which are added to $L(i+1)$. Then, the player sends one miner of $\ell(i)$ to each of newly created leaves, and decides on some active leaf $\ell'(i) \in L(i+1)$ on which to send the last miner at $\ell(i)$. After round~$i$, the cost of the game is updated as follows,
\begin{align*}
\text{Cost}(i+1) = \text{Cost}(i)+d(\ell(i),\ell'(i)),
\end{align*}
keeping track of the total distance traversed by the excess miners\footnote{This cost can slightly be refined, see \cite{cosson2023breaking}.}. A strategy for the player of the tree-mining game is called $f(k,D)$-bounded if all $k$ miners are guaranteed to attain depth $D$ with a score of at most $f(k,D)$. We have the following result,
\begin{proposition}\label{th:tm-acte}
For every $f(k,D)$-bounded strategy for the tree-mining game (TM) there exists an asynchronous collective tree exploration (ACTE) algorithm $\Bcal$ satisfying,
$$\Moves(\Bcal,k,n,D) \leq 2n + f(k,D).$$
\end{proposition}

\paragraph{Proof idea.} The proof relies on the notion of locally-greedy algorithms with target, formally defined in \cite{cosson2023breaking}. An exploration algorithm is called locally-greedy if moving robots always prefer unexplored edges over explored edges. 
In a locally-greedy algorithm with targets, each robot is assigned a target, which is a discovered node that has not been mined. When a robot is not adjacent to an unexplored edge, it simply performs one step towards its target. In the reduction to the tree-mining game, the targets correspond to the active leaves. 
When a robot mines its target, all robots sharing the same target must be re-targeted, and the additional movement cost associated to this re-targeting is bounded by the distance between the previous target and the new target. This quantity is exactly the cost in the corresponding tree-mining game. We refer to \cite{cosson2023breaking} for more details.

\begin{remark*}
In the description of the game above, leaves are not deleted but rather deactivated, and some nodes can be of degree $2$. An alternative definition of the game, which is closer from the one used in the continuous tree-mining game, would see the deactivated leaves effectively deleted from the tree, as well as nodes of degree equal to $2$ by merging their adjacent edges. This view, which is used for the reduction below, maintains a `simple' tree, as defined in the introduction.
\end{remark*}

\paragraph{Continuous Tree-Mining Game (CTM).} We defined in the beginning of Section \ref{sec:analysis} a variant of the tree-mining game, called the continuous tree-mining game (CTM). Similarly, we say that a strategy for this variant of the game is $f(k,D)$-bounded if $k$ miners are guaranteed to attain depth $D$ with a cost of at most $f(k,D)$. We shall now prove a reduction from the discrete tree-mining game to the continuous tree-mining game.

\begin{proposition}\label{th:tm-ctm}
For any $f(k,D)$-bounded strategy for the player of the continuous game (CTM), there is an $f(k,D)$-bounded strategy for the player of the discrete tree game (TM). 
\end{proposition}
\begin{proof}
The discrete strategy is defined using the continuous strategy as follows. Given the $i$-th move of the adversary, the player of the discrete game emulates the continuous game for some duration 
$t_{i+1}-t_{i}$ against some adequately defined continuous adversary. The configuration obtained in the continuous tree $T^c$ at time $t_{i+1}$ defines the response of the player on the discrete tree $T^d$ at discrete step $i$. The following properties, which highlight the correspondences between the continuous and the discrete games will always be satisfied at instants $t_1,t_2,t_3,\dots$ (a) $T^c(t_i)$ and $T^d(i)$ have the same set of vertices $V$, edges $E$ and leaves $\Lcal$. (b) For any leaf $\ell\in \Lcal:x_\ell^c=x_\ell^d$. (c) For any $v\in V\setminus \Lcal: d_v^c =d_v^d$ where $d_v^c$ (resp. $d_v^d$) denotes the distance from $v$ to $p(v)$ in $T^c$ (resp. $T^d$). (d) If for some leaf $\ell\in \Lcal$ we have $d_\ell^c<d_\ell^d$ then necessarily, $x_\ell^c =x_\ell^d =1$. 
We assume that the following properties are satisfied at time $t_{i}$ and show how we define time $t_{i+1}$ as well as the evolution of the continuous game during $[t_{i},t_{i+1}]$. 
The state of $T^c(t_{i+1})$ will then define the $i+1$-th move of the player of the discrete game by enforcing (b) in $T^d(i+1)$. Assume that the $i$-th choice of the adversary of the discrete game is a leaf $\ell(i)$ with $x_{\ell(i)}$ miners. Then, a similar move is perpetrated in $T^c$ to that same leaf, i.e. edge deletion if $x_{\ell(i)}=1$, edge elongation if $x_{\ell(i)}=2$ or fork if $x_{\ell(i)}\geq 3$, so that the property (a) is preserved. 
Then, while there is a leaf $\ell$ of $T^c$ satisfying $x_\ell\geq 2$ and $d_\ell^c <d_\ell^d$, that leaf is elongated, leading to the transfer of the excess miner from leaf to leaf. 
When that property is no longer true, we interrupt the passing of time and this defines $t_{i+1}$ as well as the state of the continuous game $T^c(t_{i+1})$. We then perform the move of the player of the discrete game which enforces (b) at step $i+1$ and note that properties (c) and (d) follow. Properties (a), (b), (c), (d) are thus all satisfied at step $i+1$. 
It is clear from the description that all moves taking place in $T^d$ up to step $i$ must have taken place in $T^c$, except for the moves inside partially completed leaves, which came at no cost in the discrete game. In conclusion the cost of the discrete (resp. continuous) game $\text{Cost}^d(i)$ (resp. $\text{Cost}^c(i)$) satisfy $\text{Cost}^d(i)\leq \text{Cost}^c(i)$.
\end{proof}

\subsection{Layered Tree Traversal} 
In this section, we first recall the problem of layered tree traversal (LTT) which is a crucial special case of layered graph traversal (LGT) \cite{fiat1998competitive}. We then recall the equivalence between fractional strategies and mixed (i.e. randomized) strategies for layered tree traversal \cite{bubeck2022shortest}. 
After, we show how a new type of guarantees for layered tree traversal can be obtained from asynchronous collective exploration (ACTE) algorithms. 


\paragraph{Layered Tree Traversal.} A searcher attempts to go from the source $r\in V$ to the target $r'\in V$ of a weighted tree $T=(V,E)$ with as little effort as possible. The searcher does not initially observe the entire tree $T$ which is instead revealed \textit{layer by layer}. The $i$-th layer corresponds to the set of nodes of combinatorial depth $i$. At step $i$, the $i$-th layer is revealed along with all the edges going from the $i-1$-th layer to the $i$-th layer. The searcher then has to choose one node of the $i$-th layer and must move to this node using previously revealed edges. The length of all edges is assumed to take its values in $\{0,1\}$. 
The goal for the searcher is to reach the last layer (which contains $r'$) while traversing the minimum number of length $1$ edges. The following quantities are useful to describe the layered tree,
\begin{itemize}
    \item the \textit{width} $w \in \Nbb$ is equal to the maximum number of nodes of a layer,
    \item the \textit{depth} $D\in \Rbb$ is equal to the shortest path distance from $r$ to $r'$,
    \item the \textit{length} $L \in \Rbb$ is equal to the sum of all edge lengths.
\end{itemize}
It was observed by \cite{fiat1998competitive} that the setting of LTT with edge lengths in $\{0,1\}$ encompasses the general setting of layered tree traversal with arbitrary edge lengths. Given an algorithm $\Acal$ traversing layered trees with edge lengths in $\{0,1\}$, we readily obtain an algorithm $\Acal'$ traversing layered trees with edge lengths in $\Nbb$, by cutting edges in segments of size $0$ and $1$, inserting intermediary layers when needed. Note that this operation does not change the quantities of interest: the width $w$, the depth $D$, and the length $L$. The argument is easily generalized to allow for edge lengths in $\{0\}\cup[\epsilon,\infty]$ where $\epsilon$ is a lower-bound on the smallest non-zero edge length, see \cite{fiat1998competitive} for full details.

\paragraph{Fractional strategies and randomized algorithms.} A \textit{deterministic algorithm} for layered tree traversal is one that maps a sequence of layers, as well as the position of the searcher in the before-last layer, to the position of the searcher in the last layer. The cost of a deterministic algorithm on some layered tree $T$ equals the total length traversed by a searcher using the algorithm to navigate in that tree, when layers are revealed one after the other. A \textit{randomized algorithm} is a probability distribution over deterministic algorithms, its cost is defined by taking the expectation of the cost over the deterministic algorithms. A \textit{fractional strategy} is one which maps a probability distribution on the before-last layer, to a probability distribution on the last layer.  The cost of a fractional strategy is equal to the sum of optimal transport costs between all consecutive probability distributions, until the last layer is revealed. The following reduction was observed by \cite{bubeck2022shortest}.
\begin{proposition}\label{prop:fractional}
    For every fractional strategy of the searcher, there is a mixed strategy incurring the same cost in expectation.
\end{proposition}

\paragraph{Proof idea.} The proof works by induction, the idea is to build a probability distribution over deterministic algorithms that has the same expected cost as the fractional strategy. At each step, the optimal transport cost between two consecutive probability distributions on both layers allows to define a coupling which tells how the position of a searcher should change, given its current position. We refer to \cite[Section 4]{bubeck2022shortest} for more details.

\paragraph{Connection to (ACTE).} We now show how an asynchronous collective tree exploration algorithm can be used to obtain a fractional strategy for layered tree traversal (LTT).

\begin{proposition}\label{th:reduction-lgt}
Given an asynchronous collective tree exploration (ACTE) algorithm $\Bcal$, for any $k\in \Nbb$, there exists a mixed strategy of the searcher $\Acal$ such that for any layered tree $T$ with edge lengths in $\{0,1\}$, we have,
\begin{equation*}    \Ebb({\normalfont\texttt{Cost}}(\Acal, T)) \leq \frac{1}{k}\normalfont\texttt{Moves}(\Bcal,k,T')
\end{equation*}
where $T'$ denotes the tree obtained from $T$ by concatenating all nodes connected by length $0$ edges.
\end{proposition}
\begin{proof}
We use the asynchronous collaborative exploration algorithm $\Bcal$ to define a fractional strategy for layered graph traversal. For this purpose we define $T'$ the tree obtained from the weighted tree $T$ by concatenating all nodes that are connected by length $0$ edges. 
We also denote by $V'_i$ the set of all nodes of $T'$ that belong to the $i$-th layer of $T$. Note that the sets $V'_i$ no longer form a partition of the nodes, since one node may belong to multiple layers. 
At each step $i$, the $i$-th layer of $T$ is revealed and we shall use our (ACTE) algorithm $\Bcal$ on $T'$ to define a fractional configuration on that layer. 

Specifically, we shall consider a run of the algorithm $\Bcal$ on $T'$ and instants $t_0\leq \dots \leq t_i\leq \dots$ such that (1) at time $t_i$, all robots are located on a node of $V_i'$ (2) before time $t_i$, no robot located on a node of $V_i'$ was ever granted a move. If these properties are satisfied for some $t_i$, it is easy to define a time $t_{i+1}$ and a sequence of robot moves such that the property will be satisfied at $t_{i+1}$. Indeed, simply grant a move to the robots until they each reach a node of $V_{i+1}'$. Since the exploration algorithm $\Bcal$ finishes in finite time, all robots will eventually reach such node. 

Note that the run defined above can be performed online, if the layer $i+1$ is always revealed at time $t_i$. Thus, this run of $\Bcal$ allows to define a fractional strategy for layered tree traversal, where the probability distribution is given by the distribution of the $k$ robots on $V'_i$, translated in a distribution over $V_i$ which is the actual $i$-th layer of $T$. It is clear that the number of robot moves between two consecutive steps is at least $k$ times the optimal transport cost between the distributions. Finally, observe that the exploration is not finished until the last layer, which we assume without loss of generality reduced to the target $\{r'\}$, has been reached by all robots. But the target is then still un-mined because no robot located at $r'$ was granted a move, thus at this instant the total number of moves of the exploration algorithm is less than $\normalfont\texttt{Moves}(\Bcal,k,T)$.  This finishes the proof of the proposition. 
\end{proof}

\section{Applications}\label{sec:applications}
In this section, we combine the results obtained for the continuous tree-mining game in Section \ref{sec:analysis} to the reduction presented in Section \ref{sec:reductions} in order to obtain new guarantees for collective tree exploration (CTE) and layered tree traversal (LTT). 
\subsection{Application to Collective Tree Exploration}
The first result that we obtain is on asynchronous collective tree exploration (ACTE).
\begin{theorem}\label{th:result-acte}
    There exists an asynchronous collective tree exploration algorithm $\Bcal$ satisfying, for any $k\in \Nbb$ and $n\geq D$,
    $$\Moves(\Bcal,k,n,D) \leq 2n + \Ocal(k^2D).$$
\end{theorem}
\begin{proof}
    Combine Proposition \ref{th:tm-ctm} and Proposition \ref{th:tm-acte} with Theorem \ref{th:main}.
\end{proof}
This result then immediately translates to the synchronous setting (CTE), providing the announced competitive ratio. 
\begin{theorem}
There exists a collective tree exploration algorithm $\Acal$ with runtime, $${\normalfont\texttt{Runtime}}(\Acal,k,n,D) \leq \frac{2n}{k}+\Ocal(kD).$$
The exploration algorithm $\Acal'$ which uses this algorithm with $k'=\sqrt{k}$ robots while leaving $k-\sqrt{k}$ robots idle at the root thus has a competitive ratio of $\Ocal(\sqrt{k})$.
\end{theorem}
\begin{proof}
    Apply Theorem \ref{th:result-acte} with Proposition \ref{th:rcte}. For the competitive ratio, observe that $$\frac{2n}{\sqrt{k}}+\text{cst}\times\sqrt{k}D \leq \text{cst}\times\sqrt{k}\left(\frac{n}{k}+D\right).$$
\end{proof}

\subsection{Application to Layered Tree Traversal}
The main result of this section is the following.
\begin{theorem}\label{th:result-ltt} There exists a collection of layered tree traversal (LTT) randomized algorithms $\{\Acal_k\}_{k\in \Nbb}$ satisfying for any layered tree $T$,
\begin{equation*}
\Ebb({\normalfont\texttt{Cost}}(\Acal_k, T)) \leq \frac{2L}{k} + \Ocal(kD),
\end{equation*}
where $L$ denotes the sum of all edge lengths in $T$ and $D$ denotes the distance from source to target.
\end{theorem}
\begin{proof}
    Apply Theorem \ref{th:result-acte} with Proposition \ref{th:reduction-lgt}.
\end{proof}
We now discuss the relevance of the above result in light of prior work on layered tree and graph traversal. The most notable difference is that our guarantee does not depend on the width $w$, contrarily to all previous work on the topic which focused on bounds of the form $\Ocal(c(w)D)$ (competitive analysis). The latest result of this kind is the recent algorithm $\Acal$ of \cite{bubeck2022shortest} satisfying, 
$\Ebb({\Cost}(\Acal, G)) \leq \Ocal(w^2D).$
Instead, our guarantees depend on the sum of all edge lengths $L$, which cannot be bounded by a function of $w$ and $D$ in full generality. We thus present below two settings illustrating the interest Theorem \ref{th:result-ltt}.

\paragraph{Unit edge lengths.} Recall that layered tree traversal can be reduced to the case of edge lengths belonging to $\{0,1\}$. We consider here the simple restriction where the edge lengths are all set equal to $1$, and where the environment is thus represented by an unweighted tree structure. This setting seems particularly relevant to robotic applications in which sensors have a given range. The unit lengths assumption entails that $D = N $, where $N$ is the number of layers, and thus that $L\leq w D$, where $w$ is the width of the layered tree. Therefore choosing $k = \floor{\sqrt{w}}$, we obtain an algorithm traversing unweighted layered trees of width $w$ with cost $\Ocal(\sqrt{w}D)$. 
To the best of our knowledge, this is the first guarantee on this problem that improves over the naive $\Ocal(L) = \Ocal(wD)$ upper bound, achieved by a simple depth-first search. 
More generally, if edges are not of unit lengths, but the ratio between the shortest and the longest edge lengths is bounded by a constant $C$, we obtain a guarantee in $\Ocal(\sqrt{Cw}D)$.

\paragraph{Average case analysis.} Another natural situation which provides a control of the value of $L$ is when the problem instance is sampled at random from probability distribution. One reasonable way to sample a layered graph of width $w$ is as follows. We first pick an arbitrary time horizon $N$ and we consider a set of layers $L_1, \dots, L_N$ that each contain $w$ nodes. One node of the last layer $L_N$ is arbitrarily chosen to be connected to the target $r'$, whereas all nodes of the first layer are connected to the source $r$. For all $i<N$, we define $E_s$ the set of edges in $L_i \times L_{i+1}$, connecting layer $i$ to layer $i+1$, by assigning to each node of $L_{i+1}$ one parent in $L_i$ chosen uniformly at random. 
Then for every edge $e\in E$, we pick an independent sample of a Bernouilli random variable $Z_e \sim \Bcal(\frac{1}{2})$ to decide whether $e$ is of length $0$ or $1$. The total length of the tree is then the sum of $|E|=wN$ independent Bernoullis, i.e. $L = \sum_{e\in E} Z_e$, and the distance from source to target is the sum of $N$ independent Bernoullis $D = \sum_{e\in P(r,r')}Z_e$ where $P(r,r')$ denotes the path from source to target. In this context, with high probability and for $N$ sufficiently large, $L$ will be of order as $\Theta(wD)$. As above, we then choose $k=\floor{\sqrt{w}}$ to obtain a guarantee in $\Ocal(\sqrt{w}D)$. This discussion could be of particular interest in the perspective of an average case analysis of layered tree and graph traversal. 

\subsection*{Acknowledgements}
The authors thank Laurent Viennot for stimulating discussions, and the Argo team at Inria. This work was supported by PRAIRIE ANR-19-P3IA-0001.
\bibliography{sample}

\appendix
\section{Proofs of Lemmas \ref{lemma:elongation-tension}, \ref{lemma:fork-tension}, \ref{lemma:fork-tension}, \ref{lemma:minoring-tension} on the dynamics of $\bx$}\label{ap:control_of_x}
\begin{lemma}\label{lemma:elongation-tension}
Consider some time interval $I=[t_1,t_2]$, in which some leaf $l$ is continuously extended. In this time interval, all moves by Algorithm \eqref{eq:dynamics} consist in moving a robot from $l$ to some other leaf of the tree, and no two moves happen simultaneously.
\end{lemma}
\begin{proof}
We denote by $\bx$ the configuration of the robots at time $t_1$ and we assume that it is stable (i.e. for any other configuration $\bz:\tau(\bx,\bz)< \OT(\bx,\bz)$. We denote by $\bx'$ the first configuration that the player will switch to after $\bx$, and by $t_3\in (t_1,t_2)$ the time at which it occurs. 
We shall show that $\bx'$ will satisfy $\bx' =\bx -\be_l+ \be_\ell$ for some leaf $\ell\neq l$. First observe that before $t_3$ the tension between any two leaves $\ell,\ell'\neq l$ in $\bx$ does not change with time (and thus remains always $<d(\ell,\ell')$). Then observe that the tension $\tau(\bx\rightarrow \bx') = \Phi_t(\bx')-\Phi_t(\bx)$ evolves continuously with time, as well as $\OT(\bx,\bx')$. So at time $t_3$, when the movement occurs, we have that $\Phi(\bx)-\Phi(\bx') = \OT(\bx,\bx') = \sum_{i\in [h]}d(\ell_i,\ell_i')$ for $\ell_1\rightarrow \ell_1',  \dots,  \ell_h\rightarrow \ell_h'$ an optimal transport plan leading from $\bx$ to $\bx'$.  Using Lemma \ref{lemma:tension-decomposition}, we get that at time $t_3$, 
$$\sum_{i\in [h]}d(\ell_i,\ell_i') = \tau(\bx\rightarrow \bx') \leq \tau(\ell_1 \rightarrow \ell_1')+\dots+\tau(\ell_h\rightarrow \ell_h'),$$
where by continuity of tensions, since no moves took place before $t_3$ we have for all $i\in [h]: \tau(\ell_i\rightarrow\ell_i')\leq d(\ell_i\rightarrow\ell_i')$, with strict inequality if $\ell_i\neq l$. This implies that $\forall i\in [h]: \ell_i = l$. Then observe that if $h>1$, the optimal transport plan has overlaps of positive length, so by the lemma above, 
$$\sum_{i\in [h]}d(\ell_i,\ell_i') = \tau(\bx\rightarrow \bx') < \tau(\ell_1 \rightarrow \ell_1')+\dots+\tau(\ell_h\rightarrow \ell_h'),$$
which is impossible using again the continuity of tensions. So it is a necessity that $h=1$ and that $\bx' = \bx-\be_l+\be_{\ell_1'}$. Now to conclude this proof, it suffices to observe that the configuration $\bx'$ is stable at time $t_3$ (i.e. for any other configuration $\bz:\tau(\bx',\bz)< \OT(\bx',\bz)$ where $\OT(\bx',\bz)$ denotes the optimal transport distance from $\bx'$ to $\bz$). 
We can then iterate the argument to show that the property will remain satisfied all the way to time $t_2$. Assume by contradiction that right after the player makes the move to $\bx'$ at $t_3$, there is another configuration $\bz\neq \bx'$ satisfying $\tau(\bx',\bz)\geq \OT(\bx',\bz)$, then the following identities held before the move at time $t_3$, $\tau(\bx,\bz) = \tau(\bx,\bx')+\tau(\bx',\bz) \geq \OT(\bx,\bx') +\OT(\bx',\bz) \geq \OT(\bx,\bz)$. So $\bz$ would have been another allowed candidate at time $t_3$. Thus, $\bz = \bx -\be_{l}+\be_{\ell'}$ for some other leaf $\ell'$. This implies that $\OT(\bx,\bx')+\OT(\bx',\bz)>\OT(\bx,\bz)$ which in turns implies that at time $t_3$, before the move of the player we had $\tau(\bx,\bz)>\OT(\bx,\bz)$, which is a contradiction.  
\end{proof}

\begin{lemma}\label{lemma:deletion-tension}
    When some leaf $l$ is deleted, denoting by $\bx$ the configuration before the deletion and by $\bx'$ the configuration right after the deletion, we have that $\forall \ell\neq l: y_\ell\geq x_\ell$. In other words, all moves by Algorithm \eqref{eq:dynamics} are from the deleted leaf to other leaves of the tree.
\end{lemma}
\begin{proof}
We denote by $\bx$ the configuration before the removal of $l$ and by $\bx'$ the configuration after the removal of $l$ and we recall that the transport distance $\OT(\bx,\bx') = \sum_{i\in [h]}d(\ell_i,\ell_i')$ where $\ell_1\rightarrow \ell_1',  \dots,  \ell_h\rightarrow \ell_h'$ is an optimal transport plan leading from $\bx$ to $\bx'$. We also recall that Lemma~\ref{lemma:tension-decomposition} gives, 
$$\tau(\bx\rightarrow \bx') \leq \tau_\bx(\ell_1 \rightarrow \ell_1')+\dots+\tau_\bx(\ell_h\rightarrow \ell_h').$$
We assume by contradiction that there is a leaf $\ell\neq l$ satisfying $y_\ell <x_\ell$. Without loss of generality, this implies that we can choose a leaf $\ell_i\neq l$ of the optimal transport plan, for $i\in [h]$, such that $y_{\ell_i'}>x_{\ell_i'}$.
Consequently the configuration $\bx'' = \bx'-\be_{\ell_i'}+\be_{\ell_i}$  where one move from $\ell_i$ to $\ell_i'$ never took place is valid, even after the deletion of leaf $l$ because obviously $\ell_i'\neq l$.
Also remember that by the proof of Lemma \ref{lemma:tension-decomposition}, 
$$\tau(\bx\rightarrow \bx')\leq \tau(\bx\rightarrow \bx'')+\tau(\ell \rightarrow \ell')$$
and recall,
    $$\OT(\bx,\bx') = \OT(\bx,\bx'')+d(\ell,\ell').$$
Since $\bx'$ was performed instead of $\bx''$, it must have been that $\tau(\bx\rightarrow \bx')-\OT(\bx,\bx')\geq \tau(\bx\rightarrow \bx'')-\OT(\bx,\bx'')$ and thus that $\tau_\bx(\ell \rightarrow \ell')\geq d(\ell,\ell')$. This is not possible, because it implies that a movement from $\ell$ to $\ell'$ would have taken place before the removal of $l$, at time $t_1^-$.
\end{proof}

\begin{lemma}\label{lemma:fork-tension}
Consider the fork of a leaf $l$ in $m$ children, and denote by $\bx$ the configuration right before the fork. We call `configuration induced by the fork' and we denote by $\bx'$ a configuration in which the $x_l$ robots formerly located at $l$ are partitioned evenly on the $m$ children of $l$, i.e. where the newly created children take values in $\{\ceil{\frac{x_l}{m}},\floor{\frac{x_l}{m}}\}$ and where $x_l'=x_l$. There exists some $\delta\in (0,1)$, such that for any fork length of $\delta$, the configuration $\bx'$ is stable after the fork. 
\end{lemma}
\begin{proof}
We recall that a fork of leaf $l$ in $m$ children leads to the creation of $m$ new leaves at distance $\delta$ from $l$. By convexity of $\Phi(\cdot)$, the configuration $\bx'$ defined above are such that for any $\ell,\ell'$ newly created children, $\tau_{\bx'}(\ell\rightarrow \ell')\leq 0$ and thus, $\tau_{\bx'}(\ell\rightarrow \ell')-d(\ell,\ell')<0$ after the fork. We consider $T(u)$ the tree corresponding to a fork length of $u\in[0,1)$. We note that for any two leaves $\ell,\ell'$, which are not both children of $l$, the value of $\tau_{\bx'}(\ell\rightarrow\ell')-d(\ell,\ell')$ evolves continuously and is known to be strictly below $0$ for $u=0$, by stability of the configuration $\bx$ before the fork. 
Therefore there is a value of $\delta>0$, such that the benefit of moving robots between leaves remains below below $0$, and by Lemma \ref{lemma:tension-decomposition}, this implies that the configuration $\bx'$ defined above is stable after the fork. 
\end{proof}

The lemmas above, also lead to the following result. 
\begin{lemma}\label{lemma:minoring-tension}
Over the course of the evolving tree game, the discrete configuration $\bx$ always satisfies $\forall \ell \in \Lcal: x_\ell\geq 1$. 
\end{lemma}
\begin{proof}
We assume by contradiction that the property is not satisfied, and we consider the first instant $t$ when a leaf was depopulated for the first time (take the supremum over the instants for which the property hold). By Lemma \ref{lemma:elongation-tension}, $t$ cannot occur in the middle of a leaf elongation because a leaf with a single robot cannot be elongated. It must thus correspond to the instant of a discrete step. But by Lemma \ref{lemma:deletion-tension} and \ref{lemma:fork-tension} no discrete steps can lead to a lonely leaf. This finishes the contradiction.
\end{proof}



\end{document}